\documentclass[11pt]{article}
 
\usepackage[utf8]{inputenc} 
\usepackage[T1]{fontenc}    
 \usepackage{url}            
 \usepackage{amsmath,amssymb,amsfonts,amsthm} 
\usepackage{nicefrac}       
   \usepackage{color,xspace}
\usepackage{algorithm,algorithmic}
  \usepackage{graphicx}
\usepackage{mdframed}
\usepackage[most]{tcolorbox}
\usepackage{empheq}
  \usepackage[letterpaper,
            colorlinks,linkcolor=ForestGreen,citecolor=ForestGreen,
            bookmarks,bookmarksopen,bookmarksnumbered]
           {hyperref}
\usepackage{vmargin}
\newtcbox{\mymath}[1][]{%
    nobeforeafter, math upper, tcbox raise base,
    enhanced, colframe=blue!30!black,
    colback=blue!30, boxrule=1pt,
    #1}

\newtheorem{problem}{Problem}
\newtheorem{lemma}{Lemma}  
\newtheorem{claim}{Claim}  
\newtheorem{theorem}{Theorem}
\newtheorem{definition}{Definition}

 \definecolor{ForestGreen}{rgb}{0.1333,0.5451,0.1333}
\setpapersize{USletter}
\setmarginsrb{.75in}{.5in}        
             {.75in}{.5in}        
             {.25in}{.25in}     
             {.25in}{.5in}      
\newlength\myindent
\newcommand\bindent{%
  \begingroup
  \setlength{\itemindent}{\myindent}
  \addtolength{\algorithmicindent}{\myindent}
}
\newcommand\eindent{\endgroup}

\let\oldbibliography\thebibliography
\renewcommand{\thebibliography}[1]{%
  \oldbibliography{#1}%
  \setlength{\itemsep}{0pt}%
}

\newcommand\myeq{\mathrel{\stackrel{\makebox[0pt]{\mbox{\normalfont\tiny def}}}{=}}}

\newcommand{\bin}[2]{\text{Bin}\left(#1,#2\right)}

\newcommand{\KL}[2]{D_{\text{KL}}({#1}||{#2} )}

\newcommand{\pr}[1]{\ensuremath{{\bf{Pr}}\left[{#1}\right]}}

\newcommand{\f}{\tilde{f}}

\newcommand{\diam}{\frac{ \log{n}}{\log{\log{n}}}}

\newcommand{\beql}[1]{\begin{equation}\label{#1}}

\newcommand{\beq}[1]{\begin{equation}\label{#1}}
\newcommand{\eeq}{\end{equation}}

\newcommand{\Prob}[1]{\ensuremath{{\bf{Pr}}\left[{#1}\right]}}
\newcommand{\Mean}[1]{\ensuremath{{\mathbb E}\left[{#1}\right]}}

\newcommand{\whp}{\textit{whp}\xspace}

\newcommand{\spara}[1]{\smallskip\noindent{\bf #1}}

\begin{document}

\title{Joint Alignment From Pairwise Differences \\ with a Noisy Oracle}
\author{
Michael Mitzenmacher\thanks{Harvard University,
    \texttt{michaelm@eecs.harvard.edu}}    ~
  Charalampos E. Tsourakakis\thanks{Boston University \& ISI Foundation,
    \texttt{ctsourak@bu.edu, babis.tsourakakis@isi.it}}
}
\maketitle

\begin{abstract} 
In this work\footnote{ {\bf Note:} This paper was accepted at the 15th Workshop on Algorithms and Models for the Web Graph (WAW 2018) and has been invited to the Journal of Internet Mathematics special issue. For state-of-the-art results on the joint alignment problem, see  \cite{chen2016projected,LMToptimal2019}.} we consider the problem of recovering $n$ discrete random variables $x_i\in \{0,\ldots,k-1\}, 1 \leq i \leq n$ (where $k$ is constant) with the smallest possible number of queries to a noisy oracle that returns for a given query pair $(x_i,x_j)$ a noisy measurement of their modulo $k$ pairwise difference, i.e., $y_{ij} = x_i-x_j  \mod k$.  This is a joint discrete alignment problem with important applications in computer vision \cite{huang2013fine,zach2010disambiguating}, graph mining \cite{tsourakakis2017predicting}, and spectroscopy imaging \cite{wang2013exact}.  Our main result is a polynomial time  algorithm that learns exactly with high probability the alignment (up to some unrecoverable offset) using $O(n^{1+o(1)})$ queries.
\end{abstract}

 \section{Introduction} 
 \label{sec:intro}

Learning a joint alignment from pairwise differences is a problem with various important applications  in computer vision  \cite{huang2013fine,zach2010disambiguating}, graph mining such as predicting signed interactions in online social networks 
\cite{tsourakakis2017predicting}, databases such as entity resolution \cite{www2020,mazumdar2016clustering,mazumdar2017clustering}, and spectroscopy imaging \cite{wang2013exact}.  Formally, there exists a set $V=[n] \myeq \{0,\ldots,n-1\}$ of $n$ discrete items labeled from $0$ to $n-1$, and an assignment $g:V \rightarrow [k] \myeq \{0,\ldots,k-1\}$ according to which each  item   is assigned one out of $k \geq 2$ possible values.  The assignment function $g$ is unknown, but we  obtain a set of corrupted  samples of the pairwise differences $\{y_{i,j} \myeq  g(i)-g(j) \mod k\}_{(i,j) \in \Omega}$, where $\Omega \subseteq [n] \times [n]$ is a symmetric index set, i.e., if $(i,j) \in \Omega$, implies $(j,i) \in \Omega$. To give an example, imagine a set of $n$ images of the same object, where each $g(i)$ is the orientation/angle of the camera  when taking the $i$-th image. The goal is to recover $g$ based on these measurements, up to some global offset $c \in [k]$ that is unrecoverable\footnote{
To see why there exists an unrecoverable offset, notice that it is impossible to distinguish the $k$ sets of inputs $ \{g(i)+c\}_{i \in V}, c\in [k]$ from all pairwise differences even if there is no noise.}.  However, learning a joint alignment from such differences  is a non-convex problem by nature, since the input space is discrete and already non-convex to begin with \cite{chen2016projected}.   
  
 \vspace{3mm} 
 
\noindent {\bf Model.} We start with the following simplified model. Later we discuss how to apply our approach to a more general noise model. Let   $0<q\leq \frac{1}{2}$, and suppose that there are $k$ groups, where $k$  is a positive constant, that we number $\{0,1,...,k-1\}$ and that we think of as being arranged modulo $k$. Let $g(u)$ refer to the group number associated with a vertex $u$.  We are allowed to query a given pair of nodes only once.  When we query a tuple of nodes $(x,y)$, we obtain an oracle answer $\tilde{f}(e)$ that is a random variable distributed according to the following equation:

\begin{equation} 
  \label{eq:model2}
    \tilde{f}(e) = \left\{\begin{array}{lr}
        g(x)-g(y) \bmod k , & \text{with probability } 1-q;\\
        g(x)-g(y)+1 \bmod k, & \text{with probability } q/2;\\
        g(x)-g(y)-1 \bmod k, & \text{with probability } q/2.
        \end{array}\right.
\end{equation}

\noindent That is, we obtain the difference between the groups when no error occurs, and with probability $q$ we obtain an error that adds or subtracts one to this gap with equal probability. According to our model, once a tuple $(x,y)$ is queried, querying $(y,x)$ is not meaningful as it will provide no new information, i.e., it will be the additive inverse of the oracle answer for $(x,y)$ modulo $k$. In this work we ask the following question: 

 \begin{tcolorbox}
 \begin{problem}
\label{joint-align-problem} 
\noindent  What is the smallest number of queries we need to perform in order to recover $g$ with high probability ({\it whp})\footnote{An event $A_n$ holds with high probability ({\it whp}) if $\lim_{n\rightarrow +\infty} \Prob{A_n}=1$.} (up to some unrecoverable global offset) under the query model described by Equation~\eqref{eq:model2}? 
\end{problem}
\end{tcolorbox}

\noindent Our main contribution is the following result, stated as Theorem~\ref{mainthm}.

\begin{theorem}
\label{mainthm} 
There exists a polynomial time algorithm that performs  $O(n^{1+o(1)})$ queries, and recovers $g$ (up to some global offset) \whp when $0<q\leq \frac{1}{2}$.
\end{theorem}

\noindent Our result extends our recent work on clustering with a faulty oracle  \cite{tsourakakis2017predicting}, and relies on techniques developed therein. Clustering with a faulty oracle  uses a querying model according to which the faulty oracle returns a noisy answer on whether two nodes belong to the same cluster or not \cite{mazumdar2017clustering,tsourakakis2017predicting}. Learning joint alignments and clustering with a faulty oracle are equivalent when $k=2$, and become different for any $k\geq 3$.  We  analyze the fundamental case $k=2$ in Section~\ref{sec:proposed}, and  then we show how our techniques extend to $k\geq 3$.  For convenience, when $k=2$ we set the two cluster ids as $\{-1,+1\}$, rather than the two possible remainders of division by 2, i.e. $\{0,1\}$, as we do for the rest of the paper.  It is worth outlining that using  the algorithm from~\cite{tsourakakis2017predicting}, cannot solve the joint alignment problem;  indeed, even with no errors, a chain of  responses on whether two nodes belong to the same cluster along a path would not generally allow us to determine whether the endpoints of a path were in the same group or not.   

Our proposed algorithm can handle queries governed by more general error models, of the form:
$$\tilde{f}(e) = g(x)-g(y)+i \mod k \text{~~~with probability~} q_i, 0 \leq i < k.$$

\noindent That is, the error does not depend on the group values $x$ and $y$, but is
simply independent and identically distributed over the values $0$ to $k-1$.    We outline how our algorithm adapts to this more general case.

\noindent {\bf Related work.} Optimal results in terms of query complexity for an even more general version of this problem that allows for $k$ to be a non-constant function, were originally given by Chen and Cand\'{e}s  \cite{chen2016projected}. In their work they study the general noise model. Again, each pair can be queried at most once, and the noisy measurement $\tilde{f}(x,y) $ is equal to 
\begin{equation*} 
    \tilde{f}(x,y) =  
        \big(g(x)-g(y)+ \eta_{xy}\big) \bmod k
\end{equation*} 
\noindent where the additive noise values $\eta_{xy} $ are i.i.d.~random variables supported on $\{0,1,\cdots,k-1\}$, with
the following probability distribution that is slightly biased towards zero for some parameter $\delta > 0$:
\begin{equation*} 
  \Prob{\eta_{xy}=i} = \left\{\begin{array}{lr}
        \frac{1}{k}+\delta, & \text{if }i=0;\\
        \frac{1}{k}-\frac{\delta}{k-1}, & \text{for each }i\neq 0.\\
         \end{array}\right.
\end{equation*}

They design an algorithm that is non-adaptive, and the underlying queries form a random binomial graph \cite{chen2016projected}, just our proposed method. Furthermore, Chen and Cand\`{e}s  prove that for  random binomial graph  query graphs,  the minimax probability of error tends to 1 if the number of queries is less than $\Omega\left (\frac{n \log n}{k \delta^2} \right )$ \cite[Theorem 2,p. 7]{chen2016projected}. Their algorithm, based on the projected power method, has a required number of queries that matches the lower bound.  Recently Larsen, Mitzenmacher, and Tsourakakis strenghtened the lower bound of Chen and Cand\'{e}s  by proving that {\em any} non-adaptive algorithm requires $\Omega( \frac{n \log n}{k \delta^2})$ queries. Furthermore, they designed an (asymptotically) optimal, simple combinatorial algorithm both in terms of query and run time complexity  \cite{LMToptimal2019}.  The results in this work are suboptimal compared to \cite{chen2016projected,LMToptimal2019}, but they use different algorithmic techniques that are possibly of independent interest.

\section{Theoretical Preliminaries} 
\label{sec:prelim}

We use the following probabilistic results for the proofs in Section~\ref{sec:proposed}.

\begin{theorem}[Chernoff bound, Theorem 2.1  \cite{janson2011random}]\label{chernoff}
  Let $X\sim\bin{n}{p}$, $\mu=np$, $a\ge 0$ and $\varphi(x)=(1+x)\ln(1+x)-x$
  (for $x\ge -1$, or $\infty$ otherwise). Then the following inequalities
  hold:
  
    \begin{eqnarray}
    \label{chernoff_ineq_low}
    \pr{X\le \mu - a} &\le& e^{-\mu\varphi\left(\frac{-a}{\mu}\right)}
    \le e^{-\frac{a^2}{2\mu}},\\
    \label{chernoff_ineq_high}
    \pr{X\ge \mu + a} &\le& e^{-\mu\varphi\left(\frac{-a}{\mu}\right)}
    \le e^{ -\frac{a^2}{2(\mu+a/3)}}.
  \end{eqnarray} 
\end{theorem}

 \noindent We define the notion of read-$k$ families, a useful concept when proving concentration results for weakly dependent variables. 
 
\vspace{0.4cm}
\begin{definition}[Read-$k$ families]
Let $X_1, \dots,   X_m$ be independent random variables. For $j \in [r]$, let $P_j \subseteq [m]$ and let $f_j$ be a Boolean function of $\{ X_i \}_{i\in P_j}$. Assume that 
$ |\{ j | i \in P_j \}| \leq k$  for every $i\in [m]$. Then, the random variables $Y_j=f_j(\{ X_i \}_{i\in P_j})$ are called a read-$k$ family.
\end{definition} 

\noindent The following result was proved by Gavinsky et al. for concentration of read-$k$ families. The intuition is that when $k$ is small, we can still obtain strong concentration results. 

\vspace{0.4cm}
\begin{theorem}[Concentration of Read-$k$ families \cite{gavinsky2014tail}] 
\label{thm:readk1}
Let $Y_1,\ldots,Y_r$ be a family of read-$k$ indicator variables with $\Prob{Y_i=1}=q$. Then for any $\epsilon>0$,
\beql{readk-upper-bound1}
\Prob{ \sum_{i=1}^r Y_i \geq (q+\epsilon) r } \leq e^{-\KL{q+\epsilon}{q} \cdot r/k} 
\eeq 
and
\beql{readk-lower-bound}
\Prob{ \sum_{i=1}^r Y_i\leq (q-\epsilon) r} \le e^{-\KL{q-\epsilon}{q} \cdot r/k}.
\eeq
\end{theorem}

\noindent Here, $D_{\text{KL}}$ is Kullback-Leibler divergence defined as 

$$ \KL{q}{p} = q \log{ \left( \frac{q}{p} \right) }+ (1-q) \log{ \left( \frac{1-q}{1-p} \right) }.$$

The following corollary of Theorem~\ref{thm:readk1} provides multiplicative Chernoff-type bounds for read-$k$ families. Notice that the parameter $k$ appears as an extra factor in denominator of  the exponent, that is why when $k$ is relatively small we  still obtain meaningful concentration results.
 
\begin{theorem}[Concentration of Read-$k$ families \cite{gavinsky2014tail}] 
\label{thm:readk}
Let $Y_1,\ldots,Y_r$ be a family of read-$k$ indicator variables with $\Prob{Y_i=1}=q$. Also, let $Y=\sum_{i=1}^r Y_i$. Then for any $\epsilon>0$,
\beql{readk-upper-bound}
\Prob{  Y \geq  (1+\epsilon) \Mean{Y} } \leq e^{-\frac{\epsilon^2 \Mean{Y}}{2k(1+\epsilon /3)} }
\eeq  

\beql{readk-lower-bound}
\Prob{  Y \leq  (1- \epsilon) \Mean{Y} } \leq e^{-\frac{\epsilon^2 \Mean{Y}}{2k} }.
\eeq  
\end{theorem}

 \section{Proposed Method} 
 \label{sec:proposed} 

\spara{Proof strategy.} Our proposed algorithm is heavily based on our work for the case $k=2$, a special case of the joint alignment problem of interest to the social networks' community \cite{tsourakakis2017predicting}.  
 According to our, we may query any pair of nodes once, and we receive the correct answer on whether the two nodes are in the same cluster, or not, with probability $1-q=\frac{1+\delta}{2}$. Here, $0<\delta<1$ is the {\em bias}.  In both cases $k=2$ and $k \geq 3$, the structure of the algorithmic analysis is identical. Let  $L= \diam$. At a high level, our proof strategy is as follows: 

\begin{enumerate}  
\item We perform $n\Delta$ queries uniformly at random. We set $\Delta = O( \frac{\log{n}}{\delta^{L}})$. 
\medskip
\item We compute the probability that a  path between $x$ and $y$ provides us with the correct information on  $g(x)-g(y)$ or not.  
\medskip
\item	 We show that there exists a large number of {\em almost edge-disjoint paths} of length $L$ between any pair of vertices with probability at least $1-\frac{1}{n^3}$.   
\medskip
\item  To learn the difference $g(x)-g(y)$ for any pair of nodes $\{x,y\}$, we take a majority vote  ($k=2$), or a plurality vote  ($k\geq 3$), among the paths we have created. A union bound in combination with (2) shows that \whp we learn $g$ up to some uknown offset. 
\medskip
\end{enumerate}

\begin{algorithm}[t]
\caption{\label{alg:2cc} Learning Joint Alignment for $k=2$} 
 \begin{algorithmic} 
\STATE $L \leftarrow \diam$
\STATE Perform $ 20n\log{n} \delta^{-L}$ queries uniformly at random.
\STATE Let $G(V,E,\f)$ be the resulting graph, $\f:E\rightarrow \{+1,-1\}$ 
 \FOR{each item pair $x,y$} 
 \STATE {$\mathcal{P}_{x,y}=\{P_1,\ldots,P_N\} \leftarrow$ Almost-Edge-Disjoint-Paths($x,y$)}  
 \STATE  $Y_i \leftarrow \prod_{e \in P_i} \f(e)$ for $i=1,\ldots,N$  
 \STATE $Y_{xy} \leftarrow \sum_{P \in \mathcal{P}_{u,v}} Y_P $
  \IF{$Y_{xy} \geq 0$} 
  \STATE {predict $g(x) = g(y)$} 
  \ELSE 
  \STATE{predict $g(x)\neq g(y)$} 
  \ENDIF
 \ENDFOR
\end{algorithmic}
\end{algorithm}

\spara{Key differences with prior work \cite{tsourakakis2017predicting}.}  While this work relies on \cite{tsourakakis2017predicting}, there are some key differences. Our main result in \cite{tsourakakis2017predicting} is that when there exist two latent clusters ($k=2$), we can recover them \whp using $O(n \log n/\delta^4)$ queries, i.e., $\Delta = O(\log n/\delta^4)$. In this work  we need to set $\Delta=O(\log n \delta^{-L})$, i.e.,  we perform a larger number of queries. Here, $L=\diam$.  An interesting open question is to reduce the number of queries when $k\geq 3$.   Since the models are different, step 2 also differs. Furthermore, the algorithm proposed in  \cite{tsourakakis2017predicting}, and the one we propose here are  different;  in \cite{tsourakakis2017predicting} we use a recursive algorithm that we analyze using Fourier analysis to get  a near-optimal result with respect to the number of queries\footnote{The information theoretic lower bound on the number of queries is $O(n\log n/\delta^2)$ \cite{hajek2016achieving}.}. Here,
we use concentration of multivariate polynomials \cite{gavinsky2014tail}, see also \cite{alonspencer,elenberg2015beyond,kim-vu,tsourakakis2011triangle}, to analyze the plurality vote of the paths that we construct between a given pair of nodes. Steps 3, 4 are almost identical both in \cite{tsourakakis2017predicting}, and here. The key difference is that our algorithm requires an average degree  $O\left (\frac{\log n}{\delta^L} \right )$ {\em only for the first level} of certain trees that we grow, for the rest of the levels a branching factor of order $O(\log n)$ suffices.

\begin{algorithm}[t]
\caption{\label{alg:edgeDisjointPaths}  Almost-Edge-Disjoint-Paths($x,y$)} 
 \begin{algorithmic}  
\STATE \textbf{Input:} $G(V,E,\f)$, $x,y \in V(G)$  \\
\STATE \textbf{Output:} Set of paths between $x,y$  \\
\STATE - Set $L \leftarrow \diam$, $\epsilon \leftarrow \frac{1}{\sqrt{\log\log{n}}}$
\STATE - Using Breadth First Search (BFS) grow a tree $T_x$  starting from $x$ as follows.  
\bindent
\STATE $\bullet$ For the first level of the tree, we choose $\frac{4\log{n}}{\delta^{L}}$ neighbors of $x$. 
\STATE  $\bullet$  For the rest of the tree we use a branching factor  $4\log{n}$ until it reaches depth equal to $\epsilon L$. 
    \eindent
 
\STATE  - Similarly, grow a tree $T_y$ rooted at $y$, node disjoint from $T_x$ of equal depth.
\STATE - Connect each leaf $x_i$ of $T_x$ to its isomorphic copy $y_i$ of $T_y$ for $i=1,\ldots,N$, where  $N$ is the total number of leaves in the two isomorphic trees:
 \bindent
\STATE  $\bullet$   From  $x_i$ of  $T_x$ (resp. $y_i$ of $T_y$),	grow  node disjoint trees until they reach depth $(\frac{1}{2}+\epsilon)L$ with branching factor $4\log{n}$.  

\STATE    $\bullet$  Finally, find an edge between $T_{x_i},T_{y_i}$  for each $i=1,\ldots,N$.
    \eindent
\STATE Return the set of constructed paths between $x,y$.
\end{algorithmic}
\end{algorithm}

\spara{An algorithm for $k=2$.} We  describe an algorithm for $k=2$, that directly generalizes to  $k\geq 3$. The caveat is that our proposed algorithm is sub-optimal with respect to the number of queries achieved in \cite{tsourakakis2017predicting}. The model for $k=2$ gets simplified to the following: let $V=[n]$ be the set of $n$ items that belong to two clusters, call them red and blue. Set $g:V \rightarrow \{\text{red},\text{blue}\}$, $R = \{v \in V(G): g(v) = \text{red} \}$ and $B = \{v \in V(G): g(v) = \text{blue} \}$, where  $0 \leq |R| \leq n$.  The function $g$ is unknown and we wish to recover the two clusters $R,B$  by querying pairs of items. (We need not recover the labels, just the clusters.)  For each query we receive the correct answer with probability $1-q=\frac{1+\delta}{2}$, where $q>0$ is the corruption probability. That is, for a pair of items $x,y$ such that $g(x)=g(y)$, with probability $q$ it is reported that  $g(x) \neq g(y)$, and similarly if $g(x) \neq g(y)$ with probability $q$ it is reported that $g(x) = g(y)$.  Since many of the lemmas in this work are proved in a similar way as in \cite{tsourakakis2017predicting}, we outline the key differences between this work and the proof in  \cite{tsourakakis2017predicting}. We prove the following Theorem. 

\begin{theorem}
\label{thm:thrm1} 
There exists a polynomial time algorithm that performs $O(\frac{n \log{n}}{ \delta^{L}})$ edge queries and recovers the clustering $(R,B)$ \whp for any gap $0< \delta < 1$.  
\end{theorem}

The pseudo-code is shown as Algorithm~\ref{alg:2cc}.  The algorithm runs over each pair of nodes, and it invokes Algorithm~\ref{alg:edgeDisjointPaths} to construct  almost edge-disjoint paths for each pair of nodes $x,y$ using Breadth First Search.   Note that since we perform $20n\log{n}\delta^{-L}$ queries uniformly at random, the resulting graph is  is asymptotically equivalent to $G \sim G(n,\frac{40\log{n}\delta^{-L}}{n})$, see \cite[Chapter 1]{frieze2015introduction}. Here, $G(n,p)$ is the classic Erd\H{o}s-R\'{e}nyi model  (a.k.a random binomial graph model) where each possible edge between each pair  $(x,y) \in {[n] \choose 2}$ is included in the graph with probability $p$ independent from every other edge.  

It turns out that our algorithm needs an average degree  $O\left (\frac{\log n}{\delta^L} \right )$ {\em only for the first level} of the trees $T_x,T_y$ that we grow from $x$ and $y$ when we invoke Algorithm~\ref{alg:edgeDisjointPaths}. For all other levels of the grown trees, we need the degree to be only $O(\log{n})$. This difference in the branching factors exists in order to ensure that the number of leaves of trees $T_x,T_y$ in  Algorithm~\ref{alg:edgeDisjointPaths} is amplified by a factor of $\frac{1}{\delta^L}$, which then allows us to apply Theorem~\ref{thm:readk}.    Using appropriate data structures, a straight-forward implementation of Algorithm~\ref{alg:2cc} runs in $O(n^2(n+m))=O(n^3\log{n}\delta^{-L})$.  
Since we use a branching factor of $O(\log n)$ for all except the first two levels of $T_x,T_y$,  we work with the $G(n,p)$ model with $p=\frac{40\log{n}}{n}$ to construct the set of almost edge disjoint paths. (Alternatively, one can think that we start with the larger random graph with more edges, and then in the construction of the almost edge disjoint paths we subsample a smaller collection of edges to use in this stage.)  The diameter of this graph \whp grows asymptotically as $L$ \cite{bollobas1998random} for this value of $p$. We use the  $G(n,\frac{40\log{n}\delta^{-L}}{n})$ model only in Lemma \ref{lem1} to prove that every node has degree at least 
$5\log{n}\delta^{-L}$.   

Recall that in the case of two clusters $\f(e) \in \{-1,+1\}$, indicating whether the oracle answers that the two endpoints of $e$ lie or not in the same cluster.  The following result follows by the fact that $\f$ agrees with the unknown clustering function $g$ on $x,y$ if the number of corrupted edges along that path $P_{xy}$ is even.

\begin{claim} 
\label{claim1} 
Consider a  path $P_{xy}$ between nodes $x,y$ of length $L$. Let $R_{xy}=\prod_{e \in P_{xy}} \f(e)$. Then, 
\begin{align}
\label{eq1} 
\Prob{R_{xy}=1|g(x)=g(y)} = \Prob{R_{xy}=-1|g(x)\neq g(y)}     =\frac{1+(1-2q)^L}{2}  = \frac{1+\delta^L}{2} &\nonumber
\end{align}
\end{claim} 

The next lemma is a direct corollary of the lower tail multiplicative Chernoff bound. 

\begin{lemma} 
\label{lem1} 
Let $G \sim G(n,\frac{40\log{n}}{\delta^Ln})$ be a random binomial graph. Then \whp all vertices have degree greater than $5\log{n}\delta^{-L}$. 
\end{lemma}

\begin{proof} 
\normalfont
The degree $deg(x)$ of a node $x \in V(G)$ follows the binomial distribution $Bin(n-1,\frac{40\log{n}}{\delta^Ln})$. Set $\gamma = \frac{3}{4}$. Then 
\begin{align*}
\Prob{  deg(x) <5 \log{n} \delta^{-L}  } &< e^{- \frac{\gamma^2}{2} 40\log{n} \delta^{-L}}  \ll n^{-1}.
\end{align*}
Taking a union bound over $n$ vertices gives the result.
\end{proof} 

We state the following key lemma, see also \cite{dudek2015rainbow,frieze2012rainbow},  that shows that we can construct for each pair of nodes $x,y$ a special type of a subgraph $G_{x,y}$.  

\begin{lemma}
\label{lem4} 
Let $\epsilon = \frac{1}{\sqrt{ \log \log n}}$, and $k=\epsilon L$.
For all pairs of  vertices $x,y \in [n]$ there exists a subgraph $G_{x,y}(V_{x,y},E_{x,y})$ of $G$ as shown in figure~\ref{fig:fig1}, {\it whp}.
The subgraph consists of two isomorphic vertex disjoint trees $T_x,T_y$ rooted at $x,y$ each of depth $k$.
$T_x$ and $T_y$ both have a branching factor of  $4\log n\delta^{-L}$ for the first level, and  $4\log n$ for the remaining levels.
If the leaves of $T_x$ are $x_1,x_2,\ldots,x_\tau,\tau\geq \delta^{-L}n^{4\epsilon/5}$ then $y_i=f(x_i)$ where $f$ is a natural 
isomorphism.
Between each pair of leaves $(x_i,y_i),i=1,2,\ldots,m$ 
there is a path $P_i$ of length  $(1+2\epsilon) L$. The paths $P_i,i=1,2,\ldots,\tau,\ldots$ are edge disjoint.
\end{lemma}

\begin{figure*}[t]
\centering
\includegraphics[width=0.5\textwidth]{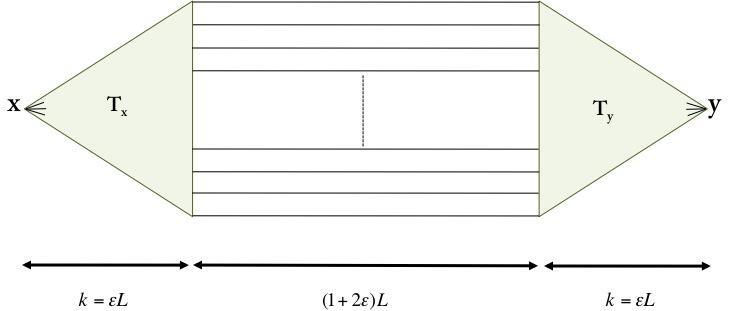} \\
\caption{\label{fig:fig1}  We create for each pair of nodes $x,y$ two node disjoint trees $T_x,T_y$ whose leaves can be matched via a natural isomorphism and linked with edge disjoint paths. For details, see Lemma~\ref{lem4}, and \cite{frieze2012rainbow,tsourakakis2014mathematical}.}
\end{figure*}

We outline that the events hold with large enough probability. For a detailed proof, please check \cite{tsourakakis2017predicting}.  The only difference with the proof of Lemma 4 in \cite{tsourakakis2017predicting} is that for the first level of  trees $T_x,T_y$, we choose $\frac{5\log n}{\delta^L}$ neighbors of $x,y$ respectively. For all other levels we use a branching factor equal to $4\log n$. The proof of Theorem~\ref{thm:thrm1} follows.  

\begin{proof}[Theorem~\ref{thm:thrm1}]
Fix a pair of nodes $x,y \in V(G)$, and suppose $x,y$ belong to the same cluster (the other case is treated in the same way).  Let $Y_1,\ldots,Y_N$  be the signs of the $N$ edge  disjoint paths connecting them, i.e., $Y_i \in \{-1,+1\}$ for all $i$. Also let $Y= \sum_{i=1}^N Y_i$. Notice that $\{Y_1,\ldots,Y_N\}$ is a read-$k$ family where $k=\frac{N}{4\log{n}\delta^{-L}}$.   By the linearity of expectation, and Lemma~\ref{lem4} we obtain 
  
  $$\Mean{Y} = N\delta^L\geq n^{\tfrac{4}{5}\epsilon} \delta^{L}.$$

\noindent By applying Theorem~\ref{thm:readk} we obtain  

\begin{align*}
\Prob{Y<0} &= \Prob{Y-\Mean{Y} < -\Mean{Y} } \leq \exp \left( -\frac{n^{4/5\epsilon} \delta^L}{\frac{2n^{4/5\epsilon}  }{4\delta^{-L}\log{n}}} \right)  =o(n^{-2}).
\end{align*}

\noindent A union bound over ${n \choose 2}$ pairs completes the proof. 
\end{proof}

\spara{Algorithm for Learning a Joint Alignment, $k\geq 3$.}  When $q= 0$, so there are no errors from $\tilde{f}(e)$, the edge queries
would allow us to determine the difference between the group numbers of vertices at the start and end of any path, and in particular would allow us
to determine if the groups were the same.  However, when $q>0$ the actual difference between the cluster ids of $x,y$, i.e., $g(x)-g(y)$ is perturbed by a certain amount of noise. In the following we discuss how we can tackle this issue. Since the proof of Theorem~\ref{mainthm} overlaps with the proof of Theorem~\ref{thm:thrm1} for $k=2$,  we outline the main differences. The idea is still the same: among the differences reported by the large number of paths we create between nodes $x,y$, the correct answer $g(x)-g(y)$ will be the plurality vote with large enough probability. The pseudocode is shown in Algorithm~\ref{alg:3cc}.

\begin{algorithm}[t]
\caption{\label{alg:3cc} Learning Joint Alignment for $k \geq 3$} 
 \begin{algorithmic} 
\STATE $L \leftarrow \diam$
\STATE Perform $ 20n\log{n} \delta^{-L}$ queries uniformly at random.
\STATE Let $G(V,E,\f)$ be the resulting graph
 \FOR{each item pair $x,y$} 
 \STATE {$\mathcal{P}_{x,y}=\{P_1,\ldots,P_N\} \leftarrow$ Almost-Edge-Disjoint-Paths($x,y$)}  
\STATE  $Y_i(x,y) \leftarrow \sum_{e \in P_i} \f(e)$ for $i=1,\ldots,N$  
 \STATE Output the plurality vote among $\{Y_1(x,y),\ldots, Y_N(x,y)\}$ as the estimate of $g(x)-g(y)$
 \ENDFOR
\end{algorithmic}
\end{algorithm}

\begin{proof}[Theorem~\ref{mainthm}]
 Let us return to the basic version of our Model, and let $X(e) \in \{-1,0,1\}$
for $e=(x,y)$ be 
$$\tilde{f}(e) - (g(x) - g(y)) \bmod k.$$
\noindent Then given a path between two vertices $x$ and $y$, 
$$g(y) = g(x) + \sum_{e \in P_{xy}} \tilde{f}(e) - \sum_{e \in P_{xy}} X(e) \bmod k.$$
Our question is now what is $Z_{xy} = \sum_{e \in P_{xy}} X(e) \bmod k$.  
We would like that
$Z_{xy}$ be (even slightly) more highly concentrated on 0 than on other
values, so that when $g(x) = g(y)$, we find that the sum of the return
values from our algorithm, 
$\sum_{e \in P_{xy}} \tilde{f}(e) \bmod k$, is most likely to be 0.
We could then conclude by looking over many almost  edge-disjoint paths that if this sum is 0
over a plurality of the paths, then $x$ and $y$ are in the same group \whp, i.e.,  the plurality value will equal g(y) - g(x) \whp.

For our simple error model, the sum $\sum_{e \in P_{xy}} X(e) \bmod k$
behaves like a simple lazy random walk on the cycle of values modulo
$k$, where the probability of remaining in the same state at each step
is $q$.  Let us consider this Markov chain on the values modulo $k$;
we refer to the values as states.  
Let $p_{ij}^{t}$ be the probability of going from state $i$
to state $j$ after $t$ steps in such a walk.  It is well known than
one can derive explicit formulas for $p_{ij}^t$;  see e.g.  \cite[Chapter XVI.2]{feller1968introduction}.  It also follows by simply finding the eigenvalues and
eigenvectors of the matrix corresponding to the Markov chain and using 
that representation.  One can check the resulting forms to determine that 
$p_{0j}^{t}$ is maximized when $j=0$, and to determine the corresponding gap
$\max_{j\in[1,k-1]} |p_{00}^{t}-p_{0j}|^{t}$.  Based on this gap, we can apply
Chernoff-type bounds as in Theorem~\ref{thm:readk} to show that the plurality of  edge-disjoint
paths will have error 0, allowing us to determine whether the endpoints of the path
$x$ and $y$ are in the same group with high probability.  

The simplest example is with $k=3$ groups, where we find
$$p_{00}^{t} = \frac{1}{3} + \frac{2}{3}\left (1-3q/2 \right )^t,$$
and 
$$p_{01}^{t} = p_{02}^{t} = \frac{1}{3} - \frac{1}{3}\left (1-3q/2 \right )^t.$$
In our case $t = L$, and we see that for any $q < 2/3$, $p_{00}^{t}$ is large enough
that we can detect paths using the same argument as for $k=2$.

For general $k$, we use that the eigenvalues of the matrix
\[
\begin{bmatrix}
    1-q & q/2 & 0 & \dots & q/2    \\
    q/2   & 1-q & q/2 & \dots & 0    \\
    \vdots & \vdots & \vdots & \vdots & \ddots \\
    q/2 & 0 & 0 & \dots  & 1-q  
\end{bmatrix}
\]
are 
$1-q+q \cos(2\pi j/k), j = 0, \ldots, k-1$, with the $j$-th corresponding 
eigenvector being $[1, \omega^j, \omega^{2j}, \ldots, \omega^{j(k-1)}]$ where
$\omega = e^{2\pi i/k}$ is a primitive $k$-th root of unity.  Here, $i$ is not an index but square root of -1, i.e., $i=\sqrt{-1}$.
In this case we have
$$p_{00}^{t} = \frac{1}{k} + \frac{1}{k}\sum_{j=1}^{k-1} \big(1-q +q \cos(2\pi j/k)\big)^t.$$
Note that $p_{00}^{t} > 1/k$.  
Some algebra reveals that the next largest value of $p_{0j}^{t}$ belongs to $p_{01}^{t}$, 
and equals
$$p_{01}^{t} = \frac{1}{k} + \frac{1}{k}\sum_{j=1}^{k-1} \omega^{-j} \big(1-q +q \cos(2\pi j/k)\big)^t.$$
We therefore see that the error between ends of a path again have the plurality value 0, with a gap of at least 
$$p_{00}^{t} - p_{01}^{t} \geq 2(1-\cos(2\pi/k))(1-q +q \cos(2\pi/k))^t.$$
This gap is constant for any constant $k \geq 3$ and $q \leq 1/2$.
\end{proof}

\noindent As we have already mentioned, the same approach could be used for the more general setting where
$$\tilde{f}(e) = g(x)-g(y)+j \text{~~with probability~~} q_j, 0 \leq j< k,$$
but now one works with the Markov chain matrix 
\[
\begin{bmatrix}
    q_0 & q_1 & q_2 & \dots & q_{k-1} \\
    q_{k-1}   & q_0 & q_1 & \dots & q_{k-2} \\
    \vdots & \vdots & \vdots &\ddots & \vdots \\
    q_1 & q_2 & q_3 & \dots & q_0
\end{bmatrix}.
\]

 \section{Conclusion} 
 \label{sec:concl} 
 
In this work we studied the problem of learning a joint alignment from pairwise differences using a noisy oracle.  Based on techniques developed in our previous work \cite{tsourakakis2017predicting}, we show how we can recover a latent alignment \whp using $O( n^{1+o(1)})$ queries. Since the time of the original publication \cite{mitzenmacher2018joint}, the key open question has been optimally  by  Larsen and the authors of the paper by providing an optimal (up to constants) non-adaptive algorithm \cite{LMToptimal2019}. An interesting open direction is to explore further adaptive algorithms for joint alignment. Finally, developing algorithms for approximate recovery is an interesting open problem. 


\end{document}